\documentclass{article}
\usepackage{amsmath,amsthm,graphicx}
\usepackage{mathptmx}

\newcommand{\br}[1]{\langle #1\rangle}
\newcommand{\Ra}{\Rightarrow}
\newcommand{\For}{\textbf{F}}
\newcommand{\Per}{\textbf{P}}
\newcommand{\RC}{\textbf{RC}}
\newcommand{\SC}{\textbf{SC}}
\newcommand{\CF}{\textbf{CF}}
\newcommand{\CS}{\textbf{CS}}
\newcommand{\Percd}{\textbf{CD(P)}}
\newcommand{\aalph}{\textit{alph}}
\newcommand{\per}{\textit{Per}}
\newcommand{\for}{\textit{For}}
\newcommand{\eps}{\lambda}

\newtheorem{theorem}{Theorem}
\newtheorem{lemma}[theorem]{Lemma}
\newtheorem{corollary}[theorem]{Corollary}
\theoremstyle{definition}
\newtheorem{definition}{Definition}
\date{}

\begin{document}
\title{Comparison of Two Context-Free Rewriting Systems with Simple Context-Checking Mechanisms}
\author{Tom\'{a}\v{s} Masopust\\
  \small Institute of Mathematics of the Czech Academy of Sciences\\[-0.8ex]
  \small {\v Z}i{\v z}kova 22, 616 62 Brno, Czech Republic\\
  \small \texttt{masopust@ipm.cz}
}

\maketitle
\begin{abstract}
  This paper solves an open problem concerning the generative power of nonerasing context-free rewriting systems using a simple mechanism for checking for context dependencies, in the literature known as semi-conditional grammars of degree $(1,1)$. In these grammars, two nonterminal symbols are attached to each context-free production, and such a production is applicable if one of the two attached symbols occurs in the current sentential form, while the other does not. Specifically, this paper demonstrates that the family of languages generated by semi-conditional grammars of degree $(1,1)$ coincides with the family of random context languages. In addition, it shows that the normal form proved by Mayer for random context grammars with erasing productions holds for random context grammars without erasing productions, too. It also discusses two possible definitions of the relation of the direct derivation step used in the literature.


\end{abstract}

\section{Introduction}
  It is well known that context-free grammars play an important role in form language theory from both practical and theoretical point of view. However, some kinds of context dependencies are required in many practical applications, such as the analysis of programming and natural languages, which, therefore, cannot be handled by context-free grammars. For that reason, some more powerful rewriting mechanisms that generate convenient proper subfamilies of the family of context sensitive languages and that make use of advantages of the simple form of context-free productions are of interest.

  This paper discusses two such rewriting mechanisms based on context-free productions. Specifically, it discusses random context grammars and their special and more simple variant, semi-conditional grammars of degree $(1,1)$. In comparison with context-free grammars where erasing productions can be eliminated without affecting the generative power, erasing productions play a significant role in random context grammars and semi-conditional grammars of degree $(1,1)$. Specifically, with them both these rewriting mechanisms characterize the family of recursively enumerable languages (see \cite{dapa} and \cite{mayer}, respectively), while without them they are less powerful then context sensitive grammars (see \cite{dapa} and \cite{paun2}, respectively). As the erasing cases of random context grammars and semi-conditional grammars of degree $(1,1)$ have been studied carefully, this paper concentrates its attention on the nonerasing variants of these grammars.

  A {\em random context grammar}, introduced by van der Walt \cite{walt} in 1970, is a context-free grammar the productions of which are applicable to a sentential form only if some of the nonterminal symbols occur in the sentential form, while some others do not. Specifically, two finite sets of nonterminal symbols---a {\em permitting} and a {\em forbidding} set---are attached to each production, and such a production is applicable to a sentential form if all permitting symbols occur in that sentential form, while no forbidding symbol does. It is well known (see \cite{bordihn94accepting,dapa}) that the family of languages generated by random context grammars is properly included in the family of context sensitive languages, and, in addition, that the elimination of either all permitting or all forbidding sets makes them less powerful (see \cite{bordihn94accepting,waltP,waltF}).

  In 1985, P{\u a}un \cite{paun2} introduced {\em semi-conditional grammars} as a variant of random context grammars, where permitting and forbidding sets are replaced with permitting and forbidding strings. According to the length of these strings, semi-conditional grammars of degree $(i,j)$, for $i,j\ge 0$, are defined. It is proved in \cite{paun2} that for any $i,j\ge 0$, the family of languages generated by semi-conditional grammars of degree $(i,j)$ contains the family of context-free languages and, in addition, is included in the family of context sensitive languages. Furthermore, semi-conditional grammars of degree $(i,j)$, where $1\le i,j\le 2$, $i\neq j$, are powerful enough to characterize the family of context sensitive languages. On the other hand, however, the precise generative power of semi-conditional grammars of degree $(1,1)$ was left open.

  This paper solves this problem so that it demonstrates that semi-conditional grammars of degree $(1,1)$ characterize the family of random context languages. As a consequence, it presents a normal form for random context grammars without erasing productions similar to the normal form for random context grammars with erasing productions proved by Mayer in \cite{mayer}, who left the question of whether this normal form also holds for random context grammars without erasing productions open. Two possible definitions of the relation of the direct derivation step used in the literature are also discussed.

  A semi-conditional grammar $G$ is called {\em simple} if for each production, either its permitting or its forbidding set is empty. It is proved in \cite{lata2009} that for every semi-conditional grammar $G$, there is an equivalent simple semi-conditional grammar $G'$ of the same degree such that $G'$ is without erasing productions if and only if $G$ is. If, in addition, $G$ is of degree $(1,1)$, terminal symbols are not contained in either permitting or forbidding sets, and the set of productions can be decomposed into two disjoint sets according to the permitting and forbidding symbols, we have so-called {\em conditional context-free rewriting systems} introduced in \cite{ita}. It is known that these rewriting systems (with or without erasing productions) are as powerful as semi-conditional grammars of degree $(1,1)$ (with or without erasing productions, respectively), see \cite{lata2009,ita}. Thus, this paper proves that they are as powerful as random context grammars. The reader is also referred to \cite{masopustJCSS} for the discussion of some additional restrictions placed on these systems.

  Finally, as far as the descriptional complexity of semi-conditional grammars is concerned, the reader is referred to \cite{lata,ipl2007,vaszil} for the latest results; an overview of these results is also presented in \cite{lata2009}. Note also that the descriptional complexity of semi-conditional grammars without erasing productions, the descriptional complexity of semi-conditional grammars of degree $(1,1)$, and the descriptional complexity of conditional context-free rewriting systems are open.

\section{Preliminaries and Definitions}
  This paper assumes that the reader is familiar with formal language theory (see \cite{salomaa}). For a set $A$, $|A|$ denotes the cardinality of $A$. For an alphabet (finite nonempty set) $V$, $V^*$ represents the free monoid generated by $V$ where the unit is denoted by $\eps$. Set $V^+ = V^*-\{\eps\}$. For a string $w\in V^*$, let $|w|$ denote the length of $w$ and $\aalph(w)$ denote the set of all symbols occurring in $w$. For a symbol $a\in V$, let $|w|_a$ be the number of occurrences of $a$ in $w$. Let $\CF$, $\CS$, ${\bf REC}$, and ${\bf RE}$ denote the families of context-free, context-sensitive, recursive, and recursively enumerable languages, respectively.

  A {\em random context grammar\/} (see~\cite{walt}) is a quadruple $G=(N,T,P,S)$, where $N$ is the alphabet of nonterminals, $T$ is the alphabet of terminals such that $N\cap T=\emptyset$, $S\in N$ is the start symbol, and $P$ is a finite set of productions of the form $(A\to x,\per,\for)$, where $A\to x$ is a context-free production, $A\in N$, $x\in V^+$ ($V=N\cup T$), and $\per,\for\subseteq N$. If for each production $(A\to x,\per,\for)\in P$, $\per=\emptyset$, then $G$ is said to be a {\em forbidding grammar}. Analogously, if for each production $(A\to x,\per,\for)\in P$, $\for=\emptyset$, then $G$ is said to be a {\em permitting grammar}.

  For two strings $u,v\in V^*$ and a production $(A\to x,\per,\for)\in P$, the relation $uAv\Ra uxv$ holds provided that
  \begin{eqnarray}\label{def1}
    \per\subseteq \aalph(uv) \textrm{\qquad and\qquad } \aalph(uv)\cap\,\for = \emptyset.
  \end{eqnarray}
  The language generated by $G$ is defined as $L(G)=\{w\in T^* : S \Ra^* w\}$, where $\Ra^*$ is the reflexive and transitive closure of the relation $\Ra$. A {\em random context language\/} is a language generated by a random context grammar. The families of languages generated by random context grammars, permitting grammars, and forbidding grammars are denoted by $\RC$, $\Per$, and $\For$, respectively. As usual, if there is no confusion, forbidding sets are omitted from the permitting productions; i.e., $(A\to x,\per)$ is written instead of $(A\to x,\per,\emptyset)$. Analogously in case of forbidding grammars.

  A {\em semi-conditional grammar of degree $(i,j)$}, for $i,j\ge 0$, is a quadruple $G=(N,T,P,S)$, where $N$ is the alphabet of nonterminals, $T$ is the alphabet of terminals such that $N\cap T=\emptyset$, $S\in N$ is the start symbol, and $P$ is a finite set of productions of the form $(A\to x,\per,\for)$, where $A\to x$ is a context-free production, $V=N\cup T$,
  \begin{enumerate}
    \item $\per\subseteq\bigcup_{k=1}^i V^k$,
    \item $\for\subseteq\bigcup_{k=1}^j V^k$,
    \item $|\per|,|\for|\le 1$,
  \end{enumerate}
  and the rewritten symbol is considered in the relation of the direct derivation step (cf. the definition~$(\ref{def1})$, where the rewritten symbol is not considered). Specifically, for two strings $u,v\in V^*$ and a production $(A\to x,\per,\for)\in P$, the relation $uAv\Ra uxv$ holds provided that
  \begin{eqnarray}\label{def2}
    \per\subseteq \aalph(uAv) \textrm{\qquad and\qquad } \aalph(uAv)\cap\,\for = \emptyset.
  \end{eqnarray}
  The language generated by $G$ is defined as $L(G)=\{w\in T^* : S \Ra^* w\}$, where $\Ra^*$ is the reflexive and transitive closure of the relation $\Ra$. A {\em semi-conditional language of degree $(i,j)$\/} is a language generated by a semi-conditional grammar of degree $(i,j)$. The family of languages generated by semi-conditional grammars of degree $(i,j)$ is denoted by $\SC(i,j)$. As usual and for the simplicity, curly brackets are omitted from the notation and $\emptyset$ is replaced with $0$; i.e., for instance, $(A\to x,p,0)$ is written instead of $(A\to x,\{p\},\emptyset)$.

  To prove the main results of this paper, we use the notion of cooperating distributed grammar systems, which are rewriting devices composed of several components represented by grammars cooperating according to a given protocol. In this paper, the considered protocol is so-called terminal derivation mode (or $t$-mode, for short) that makes the component work until it can.

  A {\em cooperating distributed $(CD)$ grammar system\/} (see \cite{gs} for more information) is a construct $\Gamma=(N,T,P_1,P_2,\dots,P_n,S)$, for some $n\ge 1$, where $N$ is the alphabet of nonterminals, $T$ is the alphabet of terminals such that $N\cap T=\emptyset$, $S\in N$ is the start symbol, and $P_1,P_2,\dots P_n$ are finite sets of productions.

  By {\em components} we understand the sets $P_i$ and by {\em $g$-components} we understand the grammars $G_i=(N,T,P_i,S)$, for all $i=1,2,\dots,n$. By a CD grammar system we understand a grammar system where all $g$-components are context-free grammars.

  A {\em permitting CD grammar system\/} (see \cite{sztaki}) is a CD grammar system where all $g$-components are permitting grammars.

  For two strings $u,v\in V^*$ ($V=N\cup T$) and a number $1\le k\le n$, let the relation $u\Ra_k v$ denote a derivation step made by the $g$-component $G_k$, and let $u\Ra_k^t v$ be a derivation such that $u\Ra_k^+ v$ and there is no $w\in V^*$ for which $v\Ra_k w$, where $\Ra_k^+$ denotes the transitive closure of the relation $\Ra_k$. The language generated by a CD grammar system $\Gamma$ working in the terminal mode ($t$\mbox{-}mode) is defined as
  \begin{eqnarray*}
    L(\Gamma) = \{w\in T^*
      & : & \textrm{ there exists } \ell\ge 1 \textrm{ such that } \alpha_i\Ra^t_{k_i} \alpha_{i+1},\\
      &   & 1\le k_i\le n,\textrm{ for each } i=1,\dots,\ell-1,\, \alpha_1=S, \textrm{ and } \alpha_\ell=w\}\,.
  \end{eqnarray*}

  Let $\Percd$ denote the family of languages generated by permitting CD grammar systems working in the $t$-mode. It is proved in \cite{sztaki} that $\Percd=\RC$. (The reader is referred to \cite{sztaki} and \cite{masopustIJFCS} for more details on CD grammar systems with permitting and forbidding components, respectively.) Finally, note that the generative power of CD grammar systems, where $g$\mbox{-}components are permitting grammars using the definition $(\ref{def2})$ of the direct derivation step, is an open problem.

\section{Results}
  Recall that it is known that $\CF\subset \SC(1,1)$ and $\RC \subset \CS$ (see, for instance, \cite{paun2} and \cite{dapa}, respectively). For an example of a semi-conditional grammar of degree $(1,1)$ generating the set of all prime numbers, the reader is referred to \cite{lata2009}.

  \subsection{Comparison of the two definitions}
  \begin{theorem}\label{thm1}
    {\bf $\SC(1,1)\subseteq\RC$}.
  \end{theorem}
  \begin{proof}
    Let $L\in\SC(1,1)$, then there is a semi-conditional grammar $G=(N,T,P,S)$ of degree $(1,1)$ such that $L(G)=L$. Construct the random context grammar $G'=(N',T,P',S)$ with $N'=N\cup \{a' : a\in T\}$ and $P'$ constructed as follows:
    \begin{enumerate}
      \item set $P'=\{(A\to h(x),h(\per),h(\for)) : (A\to x,\per,\for)\in P\}$, where $h$ is a homomorphism defined as $h(X)=X$, for $X\in N$, and $h(a)=a'$, for $a\in T$;
      \item remove each production $(A\to x,\per,\for)$ with $A\in\for$ from $P'$;
      \item replace each production $(A\to x,\per,\for)$ with $(A\to x,\per-\{A\},\for)$ in $P'$;
      \item for each $a\in T$, add $(a'\to a,\emptyset,N)$ to $P'$.
    \end{enumerate}

    Thus, $(A\to h(x),h(\per)-\{A\},h(\for))\in P'$ if and only if $(A\to x,\per,\for)\in P$ and $A\notin \for$. In addition,
    \begin{itemize}
      \item $(A\to x,\per,\for)\in P$ is applicable to $uAv$ in $G$ if and only if
      \item $\per\subseteq\aalph(uAv)$ and $\for\cap\aalph(uAv)=\emptyset$, which is if and only if
      \item $\per-\{A\}\subseteq\aalph(uv)$, $\for\cap\aalph(uv)=\emptyset$, and $A\notin\for$.
      \item This is if and only if $(A\to h(x),h(\per)-\{A\},h(\for))$ is applicable to $h(uAv)$ in $G'$.
    \end{itemize}

    As $h(\per\cup\for)\subseteq N'$, $G'$ is a random context grammar generating $L$.
  \end{proof}

  More generally, the previous proof gives a method how to transform any random context grammar using the definition (\ref{def2}) of the direct derivation step to an equivalent random context grammar using the definition (\ref{def1}). The converse transformation is proved so that each production $(A\to x,\per,\for)$ is replaced with two productions $(A\to A',\emptyset,\{X' : X\in N\})$ and $(A'\to x,\per,\for)$. Thus, both definitions of the relation of the direct derivation step are equivalent for random context grammars.

  This paper also proves the analogous result for semi-conditional grammars of degree $(1,1)$. Let {\bf $\SC'(1,1)$} denote the family of languages generated by semi-conditional grammars of degree $(1,1)$ using the definition $(\ref{def1})$, then we have the following result.

  \begin{corollary}\label{sc2->1}
    {\bf $\SC(1,1)\subseteq \SC'(1,1)$.}
  \end{corollary}
  \begin{proof}
    Modify the construction of $G'=(N',T,P',S)$ from the previous proof so that $N'=N$ and $P'$ is constructed from $P$ using only clauses 2 and 3.
  \end{proof}

  \begin{theorem}
    {\bf $\SC(1,1)=\SC'(1,1)$}.
  \end{theorem}
  \begin{proof}
    By Corollary \ref{sc2->1}, it remains to show {\bf $\SC'(1,1)\subseteq \SC(1,1)$}. Let $G=(N,T,P,S)$ be a semi-conditional of degree $(1,1)$ using the definition $(\ref{def1})$ such that $L(G)=L$. Construct the semi-conditional grammar $G'=(N',T,P',S')$ of degree $(1,1)$ using the definition $(\ref{def2})$, where $S_1$ is a new start symbol,
    $N' = N\cup\{S_1\}\cup\{[A] : A\in N\cup T\}\cup \{A' : A\in N\}\cup\{[pA],[p_1A],[p_2A] : p=(A\to\alpha,u,v)\in P\}$, and initialize \[P'= \{(S_1\to [S],0,0)\} \cup \{([a]\to a,0,0) : a\in T\}\,.\] Then, for each production $p=(A\to\alpha,u,v)\in P$, the following productions are added to $P'$.
    \begin{enumerate}
      \item $([A]\to [x]\beta,u,v)$ \qquad\quad\ for $\alpha=x\beta$, $x\in V$,
      \item[] and for each $B\in N\cup T$, add
      \item $([B]\to [pB],0,0)$,
      \item $(A\to A',[pB],A')$,
      \item $([pB]\to [p_1B],A',0)$,
      \item $([p_1B]\to [p_2B],u,v)$ \qquad for $v\neq B$,
      \item $([p_1B]\to [p_2B],0,v)$ \qquad for $u=B$ and $v\neq B$,
      \item $(A'\to \alpha,[p_2B],0)$,
      \item $([p_2B]\to [B],0,A')$.
    \end{enumerate}
    It is not hard to see that $L(G')=L(G)$.
  \end{proof}

  \subsection{Generative power}
  Recall that the following holds: $\CF\subset \SC(1,1) \subseteq \RC \subset \CS$. In the rest of this section, we prove the other inclusion, i.e., we prove that $\SC(1,1)=\RC$. To do this, we first prove two auxiliary lemmas.

  \begin{lemma}\label{lem1}
    For each random context grammar $G$, there is an equivalent random context grammar $G'$ such that $(A\to x,\per,\for)$ is a production of $G'$ implies that $A\notin\for$.
  \end{lemma}
  \begin{proof}
    Let $G=(N,T,P,S)$ be a random context grammar. Construct the random context grammar $G'=(N\cup N',T,P',S)$, where $N'=\{A':A\in N\}$ is such that $N\cap N'=\emptyset$, and $P'=\{(A\to A',\emptyset,N'),(A'\to x,\per,\for) : (A\to x,\per,\for)\in P\}$. Then, it is not hard to see that $G$ and $G'$ generate the same language and $G'$ satisfies the required property.
  \end{proof}

  The following lemma proves that every random context language is generated by a CD grammar system with permitting components working in the $t$-mode, where each permitting set is of cardinality no more than one.

  \begin{lemma}\label{lem2}
    Every random context language is generated by a permitting CD grammar system where each permitting set is either empty or a one element set.
  \end{lemma}
  \begin{proof}
    Let $L$ be a random context language, and let $G=(N,T,P,S)$ be a random context grammar generating $L$ that satisfies the property of Lemma \ref{lem1}. Let the productions of $P$ be labeled by numbers from $1$ to $n=|P|$. Then, for each labeled production $i.(A\to x,\per,\for)\in P$ with $\per=\{X_1,X_2,\dots,X_k\}$, for some $k\ge 0$, create a new component $P_i$ containing the following productions:
    \begin{enumerate}
      \item $([A,i]\to [A,i,1],\emptyset)$,
      \item $([A,i,j]\to [A,i,j+1],\{[X_j,i]\})$, for $1\le j\le k$,
      \item $([A,i,k+1]\to \br{h_i(x)},\emptyset)$, where $h_i$ is a homomorphism defined as $h_i(X)=[X,i]$, for $X\in N$, and $h_i(a)=a$, for $a\in T$,
      \item $(\br{h_i(x)}\to \br{h_i(x)},\{\br{h_i(x)}\})$,
      \item $([X,i]\to [X,i],\emptyset)$, for $X\in\for$,
      \item $([X,i]\to [X,i]',\{\br{h_i(x)}\})$, for $X\in N-\for$,
      \item $([A,i,j]\to [A,i,j],\emptyset)$, for $1\le j\le k$,
      \newcounter{k}
      \setcounter{k}{\value{enumi}}
    \end{enumerate}
    and a new component $\bar{P_i}$ containing the following productions:
    \begin{enumerate}
      \setcounter{enumi}{\value{k}}
      \item $([X,i]\to [X,j],\emptyset)$, for $X\in N$, $1\le j\le n$,
      \item $([X,k]\to [X,\ell],\{[Y,m]\})$, for $X,Y\in N$, $1\le k,\ell,m \le n$, $k\neq m$.
    \end{enumerate}
    Finally, add the component \[P_0=\{(S'\to [S,i],\emptyset), ([A,i]'\to [A,i],\emptyset), (\br{h_i(x)}\to h_i(x),\emptyset) : A\in N,\, 1\le i\le n\}\,.\]

    Let $\Gamma=(N',T,P_0,P_1,\bar{P_1},\dots,P_n,\bar{P_n},S')$ be a permitting CD grammar system, where
    \begin{eqnarray*}
      N'  & =   & \{S'\}\cup N\times \{1,2,\dots,n\}\\
          &\cup & \{[A,i,j] : i.(A\to x,\per,\for)\in P,\, 1\le j \le |\per|+1\}\\
          &\cup & \{\br{h_i(x)} : (A\to x,\per,\for)\in P,\, 1\le i\le n\}\,.
    \end{eqnarray*}

    To prove that $L(G)\subseteq L(\Gamma)$, consider a derivation step of a successful derivation of $G$. Assume that a production $(A\to x,\per,\for)\in P$ labeled by $i$ is applied in this derivation step, i.e., $uAv\Ra uxv$, $\per\subseteq\aalph(uv)$, and $\for\cap\aalph(uv)=\emptyset$. We prove that \[h_i(uAv)\Ra^t_i h_i(u)'\br{h_i(x)}h_i(v)' \Ra^t_0 h_i(uxv)\] in $\Gamma$, where $h_i(z)'$ denotes $h_i(z)$ with all nonterminal symbols primed. Furthermore, if the next production applied in $G$ is labeled by $j$, we prove that the derivation of $\Gamma$ proceeds either by productions from $P_i$, for $i=j$, or, otherwise, by productions from ${\bar{P_i}}$, i.e., $h_i(uxv)\Ra^t_{\bar{i}} h_j(uxv)$.

    Clearly, by productions from $P_i$, \[ h_i(u)[A,i]h_i(v)\Ra h_i(u)[A,i,1]h_i(v)\Ra h_i(u)[A,i,2]h_i(v) \Ra^* h_i(u)\br{h_i(x)} h_i(v)\] because all symbols from $\per$ occur in $uv$. Then, all other nonterminals can be primed since there are no symbols from $\for$ in $uv$, i.e., $h_i(u)\br{h_i(x)} h_i(v)\Ra^* h_i(u)'\br{h_i(x)} h_i(v)'$. Now, notice that only one symbol $\br{h_i(x)}$ is presented in $h_i(u)'\br{h_i(x)}h_i(v)'$, and, therefore, this component of $\Gamma$ is blocked; i.e., the whole derivation by productions from $P_i$ is $h_i(uAv)\Ra_i^t h_i(u)'\br{h_i(x)}h_i(v)'$. Then, by productions from $P_0$, the derivation proceeds as $h_i(u)'\br{h_i(x)}h_i(v)'\Ra^t_0 h_i(uxv)$. Finally, for $j=i$, productions from $P_i$ are applied again. Otherwise, if $j\neq i$, productions from $\bar{P_i}$ are applied and the derivation is $h_i(uxv)\Ra^t_{\bar{i}} h_j(uxv)$. In either case, the proof proceeds by induction.

    To prove the other inclusion, $L(\Gamma)\subseteq L(G)$, consider a successful derivation of $\Gamma$. Such a derivation is of the form $S'\Ra^t_0 \alpha_1 \Ra^t \alpha_2 \Ra^t \dots \Ra^t \alpha_k$, where $\alpha_k\in T^*$, for some $k\ge 1$. Assume that $\alpha_m \Ra^t_i \alpha_{m+1}$ by productions from $P_i$, for some $i\in\{0,j,\bar{j}\}$, where $1\le j\le n$ and $1\le m < k$, and that $\alpha_m=h_i(u_0Au_1Au_2\dots Au_r)$, for some $r\ge 0$, where $A\notin\aalph(u_0u_1\dots u_r)$, $r=0$ implies that there is no $[A,i]$ in $\alpha_m$, and $h_0\in\{h_i : 1\le i\le n\}$.

    Then, with respect to $i$:
    \begin{description}
      \item[{\bf A.}] If $i=\bar{j}$, then $\alpha_{m+1}=h_\ell(u_0Au_1Au_2\dots Au_r)$, for some $\ell\neq i$. In addition, the only applicable productions are productions from $P_\ell$ and $\bar{P_\ell}$. Therefore, the derivation proceeds as in ${\bf A}$ or ${\bf B}$.

      \item[{\bf B.}] If $i=j$, let $(A\to x,\per,\for)\in P$ be the production labeled by $i$. Then, $u_0u_1\dots u_r\in h_i((V-(\for\cup\{A\}))^*)$, which follows from the fact that the derivation is successful because if there appeared a symbol $X\in\for$ in the sentential form, the derivation would keep replacing $[X,i]$ with $[X,i]$ for ever, see production $5$. It also implies that $r\ge 1$; otherwise, there is no applicable production in $P_i$, but each component is required to make at least one derivation step. Therefore, according to the productions of $P_i$,
      \begin{eqnarray}\label{B}
        \alpha_{m+1}=h_i(u_0)'A_1h_i(u_1)'A_2h_i(u_2)'\dots A_rh_i(u_r)'\,,
      \end{eqnarray}
      where $A_1,A_2,\dots,A_r\in\{\br{h_i(x)},[A,i]',[A,i,j] : 1\le j\le n\}$, and $m+1 < k$. However, the derivation is successful only if there is no more than one occurrence of $\br{h_i(x)}$ and no occurrence of a symbol of the form $[A,i,j]$ in $\alpha_{m+1}$; otherwise, $\br{h_i(x)}$ or $[A,i,j]$ are replaced with themselves for ever, see productions $4$ and $7$. This and production $6$ imply that $A_1,A_2,\dots,A_r\in\{\br{h_i(x)},[A,i]'\}$ and $|\alpha_{m+1}|_{\br{h_i(x)}}= 1$. Finally, notice that only the productions of $P_0$ are applicable.

      Thus, we can assume that $\alpha_m=h_i(v_0Av_1)$ and $\alpha_{m+1}=h_i(v_0)'\br{h_i(x)}h_i(v_1)'$, for some $v_0v_1\in (V-\for)^*$. By productions constructed in $2$ and $5$, we have verified that $\per\subseteq\aalph(v_0v_1)$ and $\for\cap\aalph(v_0v_1)=\emptyset$. Then, \[v_0Av_1\Ra v_0xv_1\] in $G$ by the production $(A\to x,\per,\for)$.

      \item[{\bf C.}] If $i=0$, then, as shown above, there is an applicable production in $P_0$ only if $\alpha_m$ is of the form achieved in $(\ref{B})$ above, i.e., $\alpha_m=h_i(u)'\br{h_i(x)}h_i(v)'$, for some $x,uv\in V^*$, and $\alpha_{m+1}=h_i(uxv)$.
    \end{description}
    The proof now proceeds by induction.

    As $\alpha_1=[S,i]$, for some $1\le i\le n$, the inclusion is proved.
  \end{proof}

  Using the previous lemma, we can prove that any random context language is generated by a semi-conditional grammar of degree $(1,1)$.

  \begin{theorem}\label{thm3}
    {\bf $\RC\subseteq\SC(1,1)$}.
  \end{theorem}
  \begin{proof}
    Let $L$ be a random context language, and let $\Gamma=(N,T,P_1,P_2,\dots,P_n,S)$, for some $n\ge 1$, be a permitting CD grammar system working in $t$-mode generating $L$ constructed as in Lemma \ref{lem2}. Let $V=N\cup T$. Construct the semi-conditional grammar of degree $(1,1)$ as follows. For each $(A\to x,\per)\in P_i$, recall that $|\per|\le 1$, add
    \begin{enumerate}
      \item\label{1} $(S'\to [S,i],0,0)$\\
            where $1\le i\le n$;
      \item\label{2} $(A\to [x,\per],[X,i],0)$\\
            where $X\in V$;
      \item\label{3} $([x,\per]\to x,\per,0)$;
      \item\label{4} $([x,\per]\to x,[Z,i],0)$\\
            where $\per=\{Z\}$;
      \item\label{5} $([A,i]\to [x_1,i]x_2\dots x_z,\per,0)$\\
            where $x=x_1x_2\dots x_z$, for some $z\ge 1$, $x_i\in V$, $i=1,\dots,z$;
      \item\label{6} $([X,i]\to [X,Q_i],0,0)$\\
            where $X\in V$ and $Q_i=\{[x,\per] : (A\to x,\per)\in P_i\}$;
      \item\label{7} $([X,Q]\to [X,(Q-\{q\})\cup\{q'\}],0,q)$\\
            where $X\in V$, $Q\subseteq Q_i\cup Q_i'$, $Q_i'=\{x' : x\in Q_i\}$, and $q\in Q\cap Q_i$;
      \item\label{8} $([X,Q_i']\to [X,P_i],0,0)$\\
            where $X\in V$;
      \item\label{9} $([X,P]\to [X,(P-\{p_j\})\cup\{p_j'\}],0,A_j)$\\
            where $P\subseteq P_i\cup P_i'$, $P_i'=\{x' : x\in P_i\}$, $p_j$ is the label of $(A_j\to x_j,\per_j)\in P\cap P_i$, and $X\in V-\{A_j\}$;
      \item\label{10} $([X,P]\to [X,(P-\{p_j\})\cup\{p_j'\}],A_j,Y)$ and\\ $([A_j,P]\to [A_j,(P-\{p_j\})\cup\{p_j'\}],0,Y)$\\
            where $P\subseteq P_i\cup P_i'$, $p_j$ is the label of $(A_j\to x_j,\per_j)\in P\cap P_i$, $Y\in\per_j$, and $X\in V-\{Y\}$;
      \item\label{11} $([X,P_i']\to [X,j],0,0)$\\
            where $X\in V$ and $j\in\{1,2,\dots,n\}$.
      \item\label{12} $([x,P_i']\to x,0,0)$\\
            where $x\in T$;
    \end{enumerate}

    Let $G=(N',T,P',S')$ be the semi-conditional grammar of degree $(1,1)$ defined above, i.e., $P'$ is defined as described above and
    \begin{eqnarray*}
      N'=N\cup\{S'\}
          & \cup  & \{[X,i] : X\in V,\, i\in \{1,2,\dots,n\}\}\\
          & \cup  & \{[X,Q] : X\in V,\, Q\in\{Q_1,Q_2,\dots,Q_n\}, Q_i \textrm{ are defined as above}\}\\
          & \cup  & \{[X,P] : X\in V,\, P\in\{P_1,P_2,\dots,P_n\}\}\\
          & \cup  & \{[x,\per] : [x,\per]\in\bigcup_{i=1}^{n} Q_i\}.
    \end{eqnarray*}

    Informally, $G$ simulates $\Gamma$ so that it remembers the simulated component $P_i$ of $\Gamma$ in the first nonterminal, which is of the form $[X,i]$, for some $X\in V$. More specifically, productions \ref{2} to \ref{5} simulate the derivation steps of the $i$th component of $\Gamma$. Production \ref{6} starts the verification process during which none of productions \ref{2}, \ref{4}, and \ref{5} are applicable: productions constructed in \ref{7} verify that there is no symbol of the form $[x,\per]$ in the sentential form; if so, production \ref{3} is not applicable, and production \ref{8} starts to verify whether there is no applicable production in $P_i$ of $\Gamma$ (see productions constructed in \ref{9} and \ref{10}); if so, production \ref{11} changes the simulated component, or production \ref{12} finishes the derivation.

    Formally, to prove that $L(\Gamma)\subseteq L(G)$, consider a successful derivation of $\Gamma$. Such a derivation is of the form $S\Ra^t \alpha_1 \Ra^t \alpha_2 \Ra^t \dots \Ra^t \alpha_k$, where $\alpha_k\in T^*$, for some $k\ge 1$. Assume that $\alpha_m \Ra^t_i \alpha_{m+1}$ by productions from $P_i$, for some $1\le i\le n$ and $1\le m<k$. Let $\alpha_m=z_1z_2\dots z_\ell$ and $\alpha_{m+1}=y_1y_2\dots y_{\ell'}$, where $z_s,y_t\in V$ for all $s=1,2,\dots,\ell$ and $t=1,2,\dots,\ell'$. As the derivation of $G$ starts by the application of a production constructed in \ref{1}, i.e., the sentential form is of the form $[S,i]$, for some $1\le i\le n$, assume that $[z_1,i]z_2\dots z_\ell$ is the current sentential form of $G$. Then, if the rewritten symbol is the first symbol of the current sentential form of $\Gamma$, production \ref{5} is applied in $G$, and if the rewritten symbol is not the first symbol of the sentential form of $\Gamma$, production \ref{2} is applied in $G$ followed by an application of production \ref{3} or \ref{4}, where the choice depends on the permitting set. In either case, sentential forms of $\Gamma$ and $G$ modified as described above coincide except for the first symbol. However, if $x\in V$ is the first symbol of the sentential form of $\Gamma$, then $[x,i]$ is the first symbol of the sentential form of $G$, for some $1\le i\le n$. Therefore, by the corresponding derivation replacing the same symbols at the same positions as in $\Gamma$, we have that $[z_1,i]z_2\dots z_\ell\Ra^* [y_1,i]y_2\dots y_{\ell'}$ in $G$. There is no production applicable to $\alpha_{m+1}$ in $\Gamma$. Thus, production \ref{6} is applied followed by a sequence of productions constructed in \ref{7} verifying that there is no symbol of the form $[x,\per]$ in the sentential form. As there is no such symbol, production \ref{8} is applied. As no productions from $P_i$ are applicable in $\Gamma$, which means that either there is not the left-hand side of the production in the sentential form, or there is the left-hand side of the production but there is not a symbol from its permitting set in the sentential form, productions constructed in \ref{9} and \ref{10}, followed by production \ref{11}, are applicable, i.e.,
    \[\begin{array}{ccccc}
      [y_1,i]y_2\dots y_{\ell'} & \Ra   & [y_1,Q_i]y_2\dots y_{\ell'}
                                & \Ra^* & [y_1,Q'_i]y_2\dots y_{\ell'}\\
                                & \Ra   & [y_1,P_i]y_2\dots y_{\ell'}
                                & \Ra^* & [y_1,P'_i]y_2\dots y_{\ell'}\\
                                & \Ra   & [y_1,j]y_2\dots y_{\ell'}\,,
    \end{array}\]
    where $j$ is such that $\alpha_{m+1} \Ra^t_j \alpha_{m+2}$. The proof then proceeds by induction. If $m+1=k$, then production $\ref{12}$ is applied instead of production $\ref{11}$.

    To prove the other inclusion, $L(G)\subseteq L(\Gamma)$, consider a successful derivation of $G$. Such a derivation starts $S'\Ra [S,j]$, for some $1\le j\le n$. Consider a more general sentential form $[X,i]w$, for some $X\in V$, $1\le i\le n$, and $w\in (N'\cup T)^*$. To simplify the proof, denote each nonterminal symbol $[x,\per]$ by the nonterminal that has generated it. It means, if, for instance, $(A\to [x,\per],[X,i],0)$ was applied, write $[x,\per]_A$. Assume that $S\Ra^* Xf(w)$ in $\Gamma$, where $f$ is a homomorphism defined as $f([x,\per]_A)=A$, and $f(X)=X$ otherwise. Then, there are the following possibilities how to proceed the derivation:
    \begin{enumerate}
      \item If production $\ref{2}$ is applied in the successful derivation, i.e., $S'\Ra^* [X,i]uAv \Ra [X,i]u[x,\per]_Av$. Then, by the assumption, \[S\Ra^* Xf(u)Af(v) = Xf(u[x,\per]_Av)\] in $\Gamma$.

      \item\label{part2} Assume that production \ref{3} or \ref{4} is applied in the successful derivation, replacing the nonterminal $[x,\per]_A$. Then, there had to be a preceding application of a production constructed in \ref{2} in the derivation, i.e., \[S'\Ra^* [Y,i]uAv\Ra [Y,i]u[x,\per]_Av\Ra^* [X,i]u'[x,\per]_Av'\Ra [X,i]u'xv'\,,\] where $i$ is unchanged in the first nonterminals of the shown part of the derivation as proved in \ref{part4} below. By the assumption and the production $(A\to x,\per)\in P_i$, \[S\Ra^* Xf(u')Af(v') \Ra Xf(u')xf(v')\] because $\per \subseteq \aalph(Xu'v')\cap N \subseteq \aalph(Xf(u'v'))$.

      \item If production \ref{5} is applied in the successful derivation, $[X,i]w\Ra [x_1,i]x_2\dots x_kw$, then \[S\Ra^* Xf(w) \Ra x_1x_2\dots x_kf(w)\] by the production $(X\to x_1x_2\dots x_k,\per)\in P_i$.

      \item\label{part4} Finally, assume that production \ref{6} is applied in the successful derivation. Then, only productions constructed in \ref{7} and \ref{3} are applicable, followed by an application of production \ref{8}, i.e., \[[X,i]\bar{w}\Ra [X,Q_i]\bar{w}\Ra^* [X,Q_i']w\Ra [X,P_i]w\,.\] However, each of the productions constructed in \ref{7} primes a symbol $[y,\per]\in Q_i$ only if there is no nonterminal symbol $[y,\per]$ in the current sentential form. Therefore, after this part of the derivation, it is verified that $w\in V^*$, which implies that any application of a production constructed in \ref{2} is followed by an application of a production constructed in \ref{3} or \ref{4} before production \ref{8} is applied. By the assumption and the argument analogous to the argument in \ref{part2} above, \[S\Ra^* Xf(\bar{w}) \Ra^* Xw\,.\] Then, only productions constructed in \ref{9} and \ref{10} are applicable, i.e., \[[X,P_i]w\Ra^* [X,P_i']w\,.\] More specifically, if production \ref{9} is applied, then $A_j$ does not occur in the sentential form $Xw$, which implies that the production $p_j.(A_j\to x_j,\per_j)\in P_i$ is not applicable in $\Gamma$. On the other hand, if production \ref{10} is applied, then $A_j$ occurs in the current sentential form, but some $Y\in \per_j$ does not. Again, the production $p_j.(A_j\to x_j,\per_j)\in P_i$ is not applicable in $\Gamma$. As all productions of $P_i$ are checked by this part of the derivation, it is verified that there is no production in $P_i$ applicable by $\Gamma$. Then, production \ref{11} is applied, which $\Gamma$ simulates by changing the component.

      If production \ref{12} is applied, then no production constructed in $\ref{2}$ is applicable, which implies that $Xw\in T^*$, and the derivation is successfully finished.
    \end{enumerate}
    As, in all cases, the sentential form is of the form $[Y,j]w'$, for some $Y\in V$, $1\le j\le n$, and $w'\in (N'\cup T)^*$. The proof proceeds by induction.
  \end{proof}

  Thus, we have proved that the family of random context languages and the family of semi-conditional languages of degree $(1,1)$ coincide.
  \begin{corollary}
    {\bf $\RC=\SC(1,1)$}.
  \end{corollary}

  \subsection{Normal forms of random context grammars}
  This section discusses the normal forms of random context grammars. Specifically, it proves that the normal form proved by Mayer in \cite{mayer} for random context grammars with erasing productions holds for random context grammars in general. It means that it holds for random context grammars without erasing productions, too.

  \begin{definition}
    A random context grammar $G=(N,T,P,S)$ is called {\em production\mbox{-}limited} if every production from $P$ is of one of the following three forms:
    \begin{enumerate}
      \item $(A\to BC, \per,\for)$
      \item $(A\to B, \per,\for)$
      \item $(A\to a, \emptyset,\emptyset)$
    \end{enumerate}
    where $A,B,C\in N$, $a\in T$, and $\per,\for\subseteq N$.
  \end{definition}

  \begin{definition}
    A random context grammar $G=(N,T,P,S)$ is called {\em limited} if it is production-limited and, in addition, each $\per,\for\subseteq N$ is either empty or a one element set.
  \end{definition}

  Mayer \cite[Theorem 6]{mayer} proved that if erasing productions are allowed, then each recursively enumerable language can be generated by a limited random context grammar. In the nonerasing case, however, he only proved (see \cite[Lemmas 7 and 8]{mayer}) that every random context language can be generated by a production-limited random context grammar, and it was left open whether the same normal form also holds for random context grammars without erasing productions. The following corollary answers this question.

  \begin{corollary}\label{nf-cor}
    Every random context language can be generated by a limited random context grammar.
  \end{corollary}
  \begin{proof}
    Given a production-limited random context grammar, the sequence of applications of constructions of Lemma \ref{lem2}, Theorem \ref{thm3}, and Corollary \ref{sc2->1}, respectively, preserves the required form of productions. The resulting grammar is random context because there are no terminal symbols in permitting and forbidding sets. In addition, each of these sets is either empty or contains only one element.
  \end{proof}

\section{Conclusion}
  This section summarizes the results and open problems concerning random context grammars and semi-conditional grammars. In what follows, the superscript $\eps$ is added if erasing productions are allowed.

  \begin{theorem} The following holds for grammars with erasing productions. The proofs can be found in \cite{bordihn94accepting,lata2009,mayer,paun2}.
    \begin{enumerate}
      \item\label{p1} ${\bf SC^{\eps}(0,0) = CF}$.
      \item\label{p2} ${\bf CF\subset SC^{\eps}(0,1) \subseteq F^{\eps} \subset REC}$.
      \item\label{p3} ${\bf CF\subset SC^{\eps}(1,0) \subseteq P^{\eps} \subset REC}$.
      \item\label{p6} ${\bf SC^{\eps}(1,1)=RE}$.
    \end{enumerate}
  \end{theorem}

  \begin{theorem}\label{thm4} The following holds for grammars without erasing productions. The proofs can be found in \cite{bordihn94accepting,waltP,paun2,waltF}. The first part 5 is proved in this paper.
    \begin{enumerate}
      \item\label{pp1} ${\bf SC(0,0) = CF}$.
      \item\label{pp2} ${\bf CF\subset SC(0,1) \subseteq F \subset RC}$.
      \item\label{pp3} ${\bf CF\subset SC(1,0) \subseteq P \subset RC}$.
      \item\label{pp4} ${\bf SC(2,1) = SC(1,2) = CS}$.
      \item ${\bf SC(1,1) = RC \subset CS}$.
    \end{enumerate}
  \end{theorem}

  The generative power of semi-conditional grammars of degree $(0,i)$ and $(i,0)$ (with or without erasing productions), for $i\ge 2$, is not known. However, if more than one forbidding string is allowed to be attached to a production (i.e., there are sets of forbidding strings instead of only one string), it is known that such grammars (referred to as {\em generalized forbidding grammars}) are computationally complete. In addition, it is sufficient to have no more than four forbidding strings each of which is of length one or two to characterize the family of recursively enumerable languages (see~\cite[Corollary~6]{dcfs}). On the other hand, however, the question of what is the generative power of {\em generalized permitting grammars} (defined in the same manner) is an open problem.

  Let $(A\to\alpha,u,v)$ be a production of a semi-conditional grammar. If $u=v=0$, then it is said to be context-free; otherwise, it is said to be {\em conditional}. The latest descriptional complexity result showing that only a finite number of resources is needed by semi-conditional grammars to generate any recursively enumerable language is proved in \cite{ipl2007}.

  \begin{theorem}[\cite{ipl2007}]
    Every recursively enumerable language is generated by a semi-condi\-tional grammar of degree $(2,1)$ with seven conditional productions and eight nonterminals.
  \end{theorem}

  Finally, Example 4.1.1 in \cite{dapa} shows that there is no bound on the number of nonterminals for random context grammars. (The proof works for semi-conditional grammars of degree $(1,1)$ where terminals are not allowed to appear as permitting or forbidding symbols, too.) More specifically, the example shows that any random context grammar generating the language \[T_n=\bigcup_{i=1}^n \{a_i^j : j\ge 1\}\] requires, in the nonerasing case, exactly $n+1$ nonterminals and, in the erasing case, at least $f(n)$ nonterminals, for some unbounded mapping $f:N\to N$.

  In the case of semi-conditional grammars, terminal symbols are allowed to appear as both permitting and forbidding symbols. As $G=(\{S,A\},\{a_1,a_2,\dots,a_n\},P,S)$, where \[P=\{(S\to a_iA,0,0), (S\to a_i,0,0), (A\to a_iA,a_i,0), (A\to a_i,a_i,0) : 1\le i\le n\}\,,\] is a semi-conditional grammar of degree $(1,0)$ generating $T_n$, the question of whether analogous descriptional complexity results can be achieved for semi-conditional grammars of degree $(1,1)$ is open.

\bibliographystyle{plain}
\bibliography{masopust}

\begin{thebibliography}{10}

\bibitem{bordihn94accepting}
H.~Bordihn and H.~Fernau.
\newblock Accepting grammars and systems.
\newblock Technical Report 9/94, Universitat Karlsruhe, Fakultat fur
  Informatik, 1994.

\bibitem{gs}
E.~Csuhaj-Varj\'{u}, J.~Dassow, J.~Kelemen, and Gh. P{\u a}un.
\newblock {\em Grammar Systems: A Grammatical Approach to Distribution and
  Cooperation}.
\newblock Gordon and Breach Science Publishers, Topics in Computer Mathematics
  5, Yverdon, 1994.

\bibitem{sztaki}
E.~Csuhaj-Varj\'{u}, T.~Masopust, and Gy. Vaszil.
\newblock Cooperating distributed grammar systems with permitting grammars as
  components.
\newblock {\em Romanian Journal of Information Science and Technology},
  12(2):175--189, 2009.

\bibitem{dapa}
J.~Dassow and Gh. P\u{a}un.
\newblock {\em Regulated Rewriting in Formal Language Theory}.
\newblock Springer-Verlag, Berlin, 1989.

\bibitem{waltP}
S.~Ewert and A.~P.~J. van~der Walt.
\newblock A pumping lemma for random permitting context languages.
\newblock {\em Theoretical Computer Science}, 270(1--2):959--967, 2002.

\bibitem{masopustJCSS}
T.~Masopust.
\newblock Simple restriction in context-free rewriting.
\newblock Submitted manuscript.

\bibitem{lata2009}
T.~Masopust.
\newblock A note on the generative power of some simple variants of
  context-free grammars regulated by context conditions.
\newblock In A.H. Dediu, A.M. Ionescu, and C.~Mart{\' i}n-Vide, editors, {\em
  LATA 2009 proceedings}, volume 5457 of {\em Lecture Notes in Computer
  Science}, pages 554--565. Springer-Verlag, 2009.

\bibitem{masopustIJFCS}
T.~Masopust.
\newblock On the terminating derivation mode in cooperating distributed grammar
  systems with forbidding components.
\newblock {\em Internation Journal of Foundations of Computer Science},
  20(2):331--340, 2009.

\bibitem{dcfs}
T.~Masopust and A.~Meduna.
\newblock Descriptional complexity of generalized forbidding grammars.
\newblock In {\em Proceedings of 9th International Workshop on Descriptional
  Complexity of Formal Systems}, pages 170--177. High Tatras, Slovakia, 2007.

\bibitem{lata}
T.~Masopust and A.~Meduna.
\newblock Descriptional complexity of grammars regulated by context conditions.
\newblock In {\em Pre-proceedings of 1st International Conference on Language
  and Automata Theory and Application (LATA 2007)}, pages 403--411, Tarragona,
  Spain, 2007.

\bibitem{ipl2007}
T.~Masopust and A.~Meduna.
\newblock Descriptional complexity of semi-conditional grammars.
\newblock {\em Information Processing Letters}, 104(1):29--31, 2007.

\bibitem{ita}
T.~Masopust and A.~Meduna.
\newblock On context-free rewriting with a simple restriction and its
  computational completeness.
\newblock {\em RAIRO -- Theoretical Informatics and Applications},
  43(2):365--378, 2009.

\bibitem{mayer}
O.~Mayer.
\newblock Some restrictive devices for context-free grammars.
\newblock {\em Information and Control}, 20:69--92, 1972.

\bibitem{paun2}
Gh. P\u{a}un.
\newblock A variant of random context grammars: Semi-conditional grammars.
\newblock {\em Theoretical Computer Science}, 41:1--17, 1985.

\bibitem{salomaa}
A.~Salomaa.
\newblock {\em Formal languages}.
\newblock Academic Press, New York, 1973.

\bibitem{walt}
A.~P.~J. van~der Walt.
\newblock Random context grammars.
\newblock In {\em Proceedings of the Symposium on Formal Languages}, pages
  163--165. 1970.

\bibitem{waltF}
A.~P.~J. van~der Walt and S.~Ewert.
\newblock A shrinking lemma for random forbidding context languages.
\newblock {\em Theoretical Computer Science}, 237(1-2):149--158, 2000.

\bibitem{vaszil}
Gy. Vaszil.
\newblock On the descriptional complexity of some rewriting mechanisms
  regulated by context conditions.
\newblock {\em Theoretical Computer Science}, 330:361--373, 2005.

\end{thebibliography}
\end{document}